\documentclass[sigconf, anonymous=false]{acmart}

\usepackage{graphicx}
\usepackage{balance}
\usepackage{amsmath,amsfonts,amsthm,epsfig,epstopdf,url,array}
\usepackage{amsthm}
\usepackage{float}
\usepackage{algorithm}
\usepackage[noend]{algpseudocode}
\usepackage{xspace}

\theoremstyle{remark}
\newtheorem{observation}{Observation}[section]
\newtheorem{definition}{Definition}[section]
\newtheorem{problem}{Problem}[section]
\newtheorem{proposition}{Proposition}[section]
\newtheorem{lemma}{Lemma}[section]

\newcommand{\Sup}{\text{sup}}
\newcommand{\InnermostTruss}{\mathit{innermostTruss}}

\AtBeginDocument{%
  \providecommand\BibTeX{{%
    \normalfont B\kern-0.5em{\scshape i\kern-0.25em b}\kern-0.8em\TeX}}}

\setcopyright{rightsretained}
\copyrightyear{2020}
\acmYear{2020}
\acmDOI{}

\acmConference[MLG '20]{MLG '20: 16th International Workshop on Mining and Learning With Graphs}{August 24, 2020}{San Diego, CA, USA}
\acmBooktitle{MLG '20: 16th International Workshop on Mining and Learning With Graphs, August 24, 2020, San Diego, CA, USA}
\acmPrice{}
\acmISBN{}

\begin{document}


\title{Efficient Algorithms to Mine Maximal Span-Trusses From Temporal Graphs}


\author{Quintino F. Lotito}
\affiliation{%
  \institution{University of Trento}
  \city{Trento}
  \country{Italy}}
\email{quintino.lotito@studenti.unitn.it}

\author{Alberto Montresor}
\affiliation{%
  \institution{University of Trento}
  \city{Trento}
  \country{Italy}}
\email{alberto.montresor@unitn.it}

\renewcommand{\shortauthors}{Lotito and Montresor, et al.}

\begin{abstract}
    Over the last decade, there has been an increasing interest in temporal graphs, pushed by a growing availability of temporally-annotated network data coming from social, biological and financial networks.
    
    Despite the importance of analyzing complex temporal networks, there is a huge gap between the set of definitions, algorithms and tools available to study large static graphs and the ones available for temporal graphs.
    
    An important task in temporal graph analysis is mining dense structures, i.e., identifying high-density subgraphs together with the span in which this high density is observed. 
    
    In this paper, we introduce the concept of $(k, \Delta)$-truss (span-truss) in temporal graphs, a temporal generalization of the $k$-truss, in which $k$ captures the information about the density and $\Delta$ captures the time span in which this density holds. We then propose novel and efficient algorithms to identify maximal span-trusses, namely the ones not dominated by any other span-truss neither in the order $k$ nor in the interval $\Delta$, and evaluate them on a number of public available datasets.
\end{abstract}

\ccsdesc[500]{Theory of computation~Graph algorithms analysis}
\ccsdesc[500]{Mathematics of computing~Graph theory}
\ccsdesc[500]{Information systems~Spatial-temporal systems}
\ccsdesc[300]{Information systems~Data mining}
\ccsdesc[300]{Information systems~Web searching and information discovery}

\keywords{Community detection, Dense structures, Graph mining, Social networks analysis, Temporal graphs}

\maketitle

\section{Introduction}
Despite the fact that graph theory has been studied for centuries, in the last years there has been an explosion in the interest of the research community in network-related fields. This is mainly motivated by the increasing interest in social networks -- which can be defined as a set of social entities (such as people, groups, and organizations) together with the relationships or interactions between them -- and by a proliferating availability of network datasets coming from online social networks (e.g., Facebook, Twitter, Instagram, YouTube), biological networks (e.g., molecular interactions) or financial interactions.  

So far, most of the work in social network analysis has focused on static graphs. The growing availability of temporally-annotated network data coming from social, biological and financial networks creates the opportunity to fill the gap between the set of definitions, algorithms and tools available for large static graphs, and the ones available to analyze temporal graphs. The latter are defined as graphs that change over time (i.e., whose edges are not continuously active). However, it is not yet clear how introducing the notion of time will affect the computational complexity of combinatorial graph problems~\cite{Holme}. 

\begin{figure}
\centering
\includegraphics[width=\linewidth]{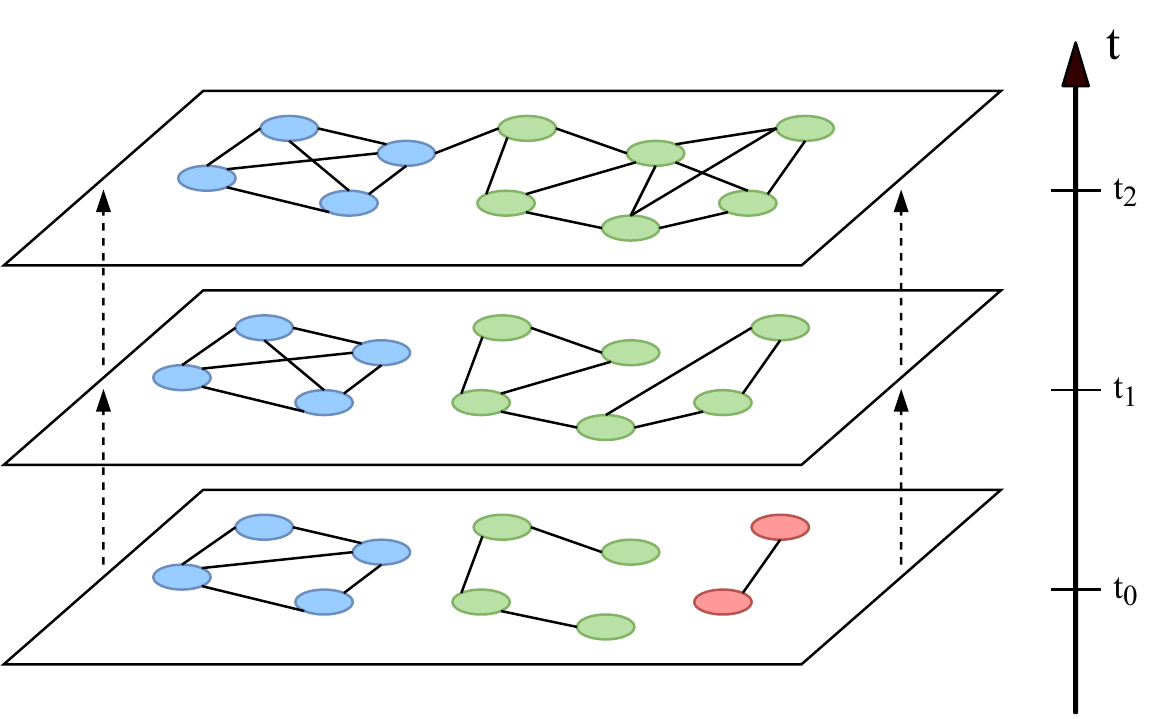}
\caption{A temporal graph with time-evolving communities. It is represented as a sequence of static graphs; each static graph is a snapshot of the temporal graph at a certain time.}
\end{figure}

Just to mention a few examples, temporal graph modelling and analysis of temporal properties can have applications in sociology and social network analysis (e.g., find voting patterns based on social media posts); security and distributed computing (e.g., design strategy to contain the spread of malware in computing devices); biology (e.g., study the set of chemical reactions that occur in a healthy organisms)~\cite{Holme}.

A property of real-world graphs is that they tend to be globally sparse but locally dense, meaning that while the entire graph is sparse (i.e., vertices have a small average degree), it contains dense subgraphs (i.e., groups of vertices with a large number of links among each other). In general, density is an indication of relevance. Dense regions in a network may indicate high degrees of interaction and mutual similarity. In real-world applications, these regions may indicate characteristics like attractive forces or favourable environments~\cite{survey}.

The enumeration of the dense components of a graph can either be the main goal of an analysis task, or act as a preprocessing step aiming to reduce the graph by removing sparse parts, in order to conduct more complex and time-consuming analysis~\cite{ChangLi}.

A number of definitions of dense structures have been proposed in literature, ranging from cliques (i.e., subgraphs in which every vertex is adjacent to every other vertex), to some relaxations of the clique, such as $k$-cores, the $k$-trusses, or the $k$-plexes.

The previously mentioned concepts of dense structures can be generalized to the temporal case, in which one can be interested in mining high-density subgraphs together with the span in which this high density is observed. Having a set of tools to extract these structures enables a detailed comprehension of the network dynamics and can act as a building block towards more complex tasks and applications~\cite{cores}. 

To name some examples of applications, we can rely on temporal dense structures computation to mine stories from social networks (i.e., events capturing popular attention in social media), which can be identified by finding a group of entities (i.e., people, locations, companies or products) strongly associated for a reasonable amount of time~\cite{social}; we can mine well-acquainted individuals from a collaboration network and form successful teams; we can analyze protein-interaction networks and locate protein complexes that are densely interacting at different states, indicating possible underlying regulatory mechanisms~\cite{bff}.

In this paper, we follow the approach of Galimberti et al.~\cite{cores}, who introduced the concept of the span-cores of a temporal graph (a temporal generalization of the $k$-core dense structure), and define the concept of $(k, \Delta)$-\textit{trusses} (\textit{span-trusses}), a temporal generalization of the $k$-truss, in which $k$ captures the information about the density and $\Delta$ captures the time span in which this density holds. We propose novel and efficient algorithms to discover the maximal span-trusses of a temporal graph, i.e., the ones not dominated by any other span-truss neither in the order $k$ nor in the interval $\Delta$. 

We conclude the paper by evaluating our contributions on a number of public available real-world network datasets, showing that our proposals consistently outperform the baseline proposed for this task.
\section{Background}
$k$-truss is a dense structure which considers the involvement between the structures of edges and triangles. It has been introduced based on the observation of social cohesion, where triangles play an essential role~\cite{Cohen}. The $k$-truss community model has three significant advantages: strong guarantee on cohesive structure, few parameters and low computational cost~\cite{truss_dyn}.

\begin{definition}[Triangle]
Given a graph $G=(V,E)$, a \textit{triangle} in $G$ is a cycle of length 3.
\end{definition}

\begin{definition}[Support of an edge]
Given a graph $G=(V,E)$ and an edge $e \in E$, the \textit{support} $\sup(e)$ is the number of triangles that $e$ participates in.
\end{definition}

\begin{definition}[$k$-truss]
Given a graph $G=(V,E)$, the $k$-\textit{truss} of $G$, where $k \geq 2$, is defined as the largest subgraph $g$ of $G$ in which every edge is contained in at least $(k-2)$ triangles within the subgraph, i.e., $\sup_g(e) \geq k-2$, $\forall e \in g$. 
\end{definition}

It is easy to see that a $k$-truss is an \emph{edge-induced subgraph}.

\begin{definition}[Maximal $k$-truss]
A $k$-truss $T_k$ of a graph $G$ is said to be \textit{maximal} if there does not exist any other $k$-truss $T_{k'}$ such that $k' > k$.
\end{definition}

\begin{problem}[Truss decomposition]
The problem of truss decomposition in a graph $G$
is to find the (non-empty) $k$-trusses of $G$ for all $k$~\cite{truss}. 
\end{problem}

\begin{figure}
\centering
\includegraphics[width=\linewidth]{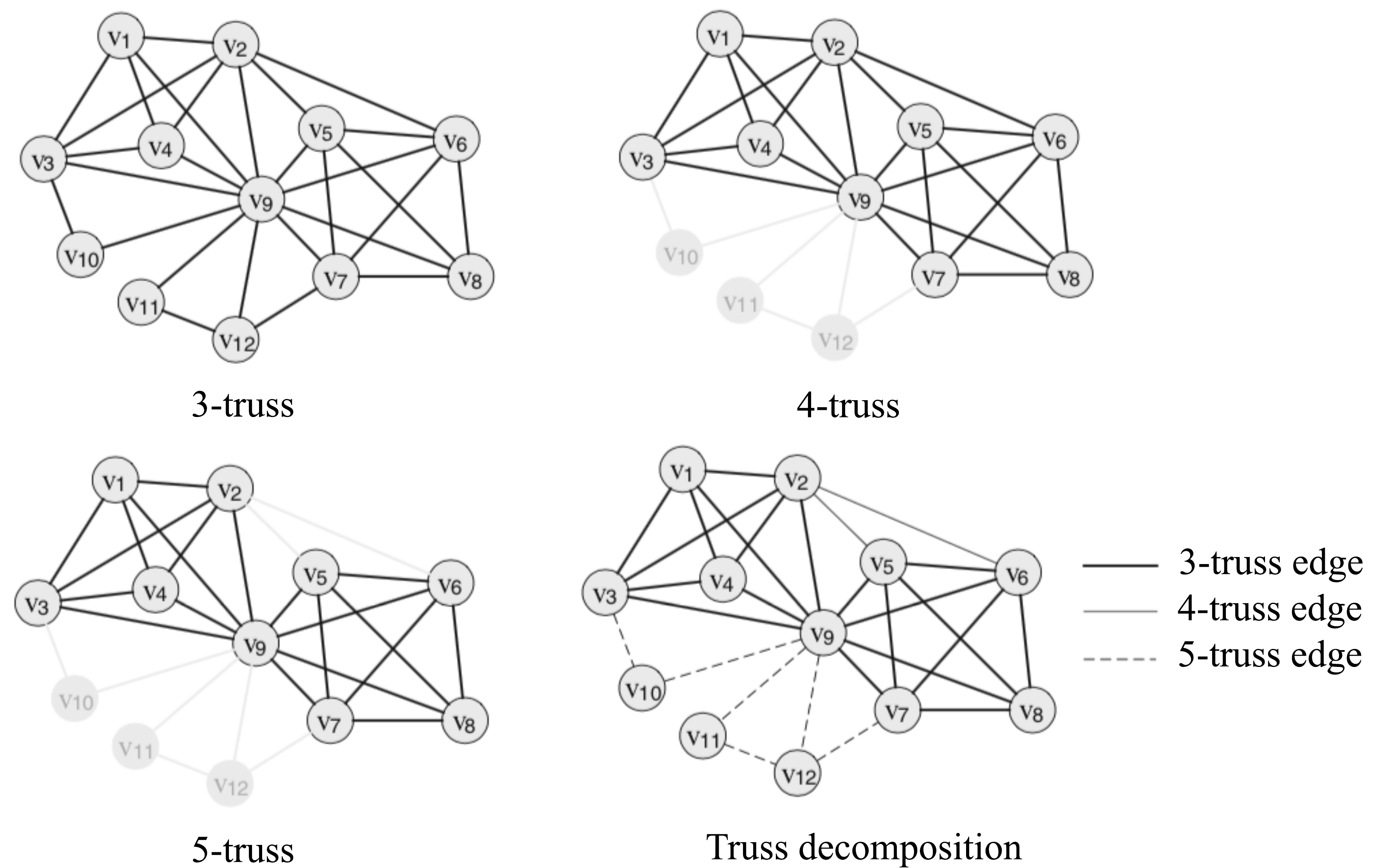}
\caption{Example of the truss decomposition of a graph~\cite{ChangLi}.}
\label{label_decomposition}
\end{figure}

\begin{observation}[Containment]
Each $k$-truss of a graph $G$ is a subgraph of the $(k-1)$-truss of $G$; for example, in~\ref{label_decomposition}, the $5$-truss is a subgraph of the $4$-truss which in turn is a subgraph of the $3$-truss.
\end{observation}

An algorithm to efficiently compute the truss decomposition of a static, unweighted, undirected graph $G=(V,E)$ has been proposed by Wang et al.~\cite{truss}. This algorithm resorts to an in-memory triangle counting algorithm~\cite{schank} and bin sort to achieve a complexity of $O(|E|^{1.5})$.
\section{Problem statement}
We are given a temporal graph $G=(V,T,\tau)$, where $V$ is a set of vertices, $T = [0,1,...,t_{max}] \subseteq \mathbb{N} $ is a discrete-time domain, and $\tau : V \times V \times T \rightarrow \{0, 1\}$ is a function defining for each pair of vertices $u,v \in V$ and each timestamp $t \in T$ whether edge $(u,v)$ exists in $t$.  

We denote $E = \{(u,v,t)\;|\;\tau(u,v,t) = 1\}$ the set of all temporal edges. Given a timestamp $t \in T$, the set of edges existing at time $t$ is $E_t = \{(u,v) \; | \;\tau(u,v,t) = 1\}$. 

A temporal interval $\Delta = [ t_s , t_e ]$ is contained into another temporal interval $\Delta'  = [t'_s,t'_e]$, denoted $\Delta  \sqsubseteq \Delta'$, if $t'_s \leq t_s$ and $t'_e \geq t_e$. 

Given an interval $\Delta \sqsubseteq T$, we denote $E_\Delta = \cap_{t \in \Delta} E_t$ the edges existing in all timestamps of $\Delta$. 
Given an interval $\Delta \sqsubseteq T$, we denote $G_\Delta = (V,E_\Delta)$ as the static graph with vertices V and edges $E_\Delta$. 

We define the temporal support of an edge $e$ over the temporal interval $\Delta$ to be equal to the support on the graph $G_\Delta$, denoted as $\Sup_\Delta(e)$.

\begin{definition}[($k,\Delta$)-truss]
The ($k,\Delta$)-truss or \textit{span-truss} of a temporal graph $G = (V, T, \tau)$ is the largest subgraph of $G_\Delta$ in which every edge is contained in at least $(k-2)$ triangles within the subgraph, i.e, $\Sup_\Delta(e) \geq k-2$, where $\Delta \sqsubseteq T$ is a temporal interval and $k \geq 2$. We will often denote the ($k,\Delta$)-truss as $T_{k, \Delta}$. 
\end{definition}

A ($k,\Delta$)-truss is a dense subgraph (where $k$ is the cohesiveness constraint), together with its temporal span, i.e., the span $\Delta$ for which the subgraph satisfies the cohesiveness constraint.

\begin{problem}[Span-truss decomposition]
Given a temporal graph $G$, find the set of all $(k, \Delta)$-trusses of $G$.
\end{problem}

\begin{observation}
For a fixed temporal interval $\Delta \sqsubseteq T$, finding all span-trusses that have $\Delta$ as their span is equivalent to computing the classic truss decomposition of the static graph $G_\Delta = (V,E_\Delta)$.
\label{lentezza}
\end{observation}

Similarly to what has been proved for the span-cores~\cite{cores}, the total number of span-trusses may be too large for human inspection. In fact, the total number of temporal intervals contained in the whole time domain $T$ is $\frac{|T|(|T|+1)}{2}$, so the total number of span-trusses is $O(|T|^2 \times k_{\text{max}} )$, where $k_{\text{max}}$ is the largest value of $k$ for which a ($k,\Delta$)-truss exists. For this reason, it is worthwhile to focus only on the most important trusses, the maximal ones, as defined next.

\begin{definition}[Maximal span-truss]
A span-truss $T_{k, \Delta}$ of a temporal graph $G$ is said to be \textit{maximal} if there does not exist any other span-truss $T_{k',\Delta'}$ of $G$ such that $k \leq k'$ and $\Delta \sqsubseteq \Delta'$.
\label{maximal_spantruss}
\end{definition}

A span-truss is recognized as maximal if it is not dominated by another span-truss both on order $k$ and the span $\Delta$.
In our temporal setting, the number of maximal span-trusses is $O(|T |^2)$, as, in the worst case, there may be one maximal span-truss for every temporal interval. However, similarly to the maximal span-cores, we expect the number of maximal span-trusses to be much smaller.

\begin{problem}[Maximal span-truss mining]
Given a temporal graph $G$, find the set of all maximal $(k, \Delta)$-trusses of G.
\label{maximal_problem}
\end{problem}

We now outline and prove some properties which will be useful later.

\begin{proposition} [Span-truss containment]
For any two span-trusses $T_{k, \Delta}$, $T_{k',\Delta'}$ of a temporal graph $G$, it holds that $k' \leq k \land \Delta' \sqsubseteq \Delta \implies T_{k, \Delta} \subseteq T_{k', \Delta'}$.
\label{prop_containment}
\end{proposition}

\begin{proof}
The result can be proved by separately showing that (i) $k' \leq k \implies T_{k, \Delta} \subseteq T_{k', \Delta}$, and (ii) $\Delta' \sqsubseteq \Delta \implies T_{k, \Delta} \subseteq T_{k, \Delta'}$.
\newline
(i) holds because every $e \in E_\Delta$ is in at least $k$ triangles in the subgraph $T_{k, \Delta}$, thus every $e$ is also in at least $k'$ triangles since $k' \leq k$; this means that $T_{k, \Delta} \subseteq T_{k', \Delta}$. 
\newline
(ii) holds because $\Delta' \sqsubseteq \Delta \implies E_\Delta \subseteq E_{\Delta'} \implies \forall e \in E_\Delta, e \in E_{\Delta'}$. If $e$ is in at least $k$ triangles in $T_{k, \Delta}$ then it is in at least $k$ triangles also in $T_{k, \Delta'}$, so $T_{k, \Delta} \subseteq T_{k, \Delta'}$.
\end{proof}

\begin{figure}
\centering
\includegraphics[width=\linewidth]{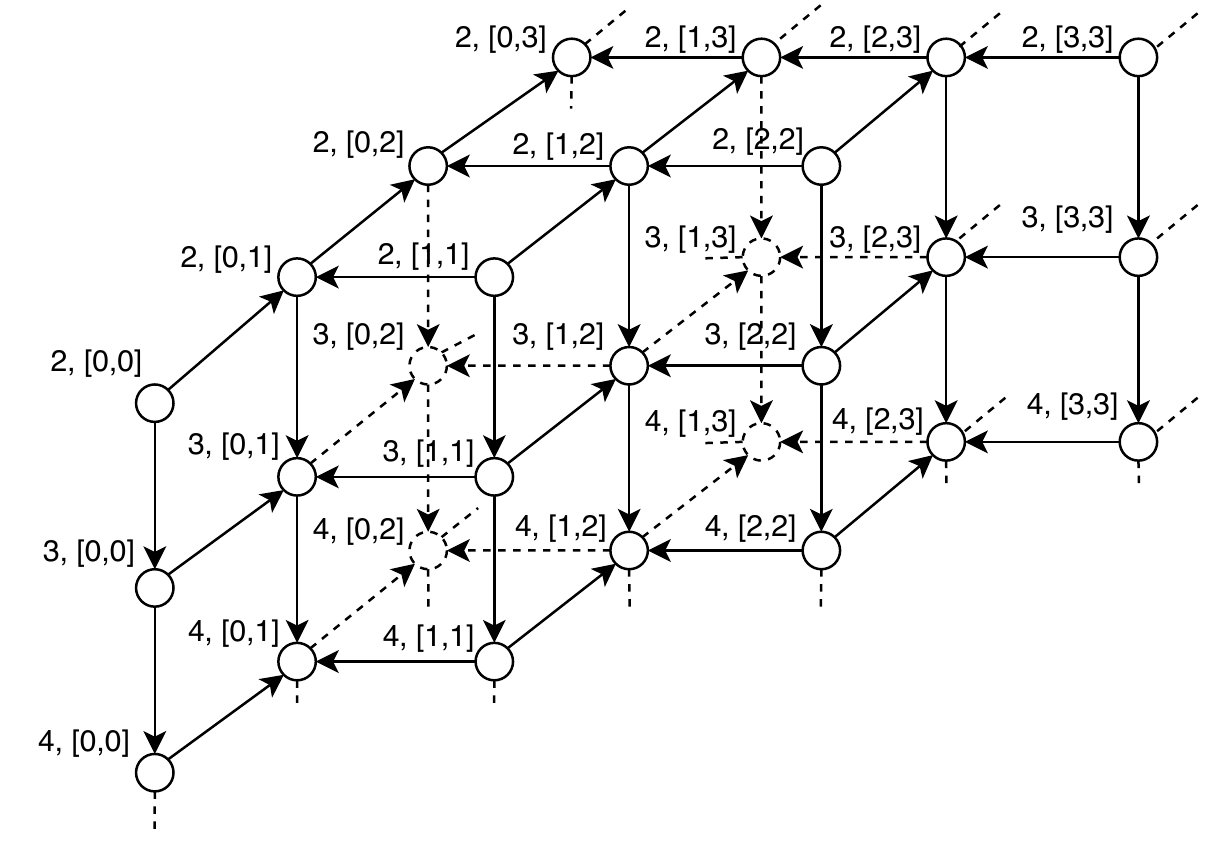}
\caption{Graphical representation of the containment property. Span-trusses follow the same structure of span-cores. For a temporal span $\Delta = [t_s , t_e]$, the $(k, \Delta)$-truss is depicted as a node labeled "$k, [t_s , t_e]$", an arrow $T_1 \rightarrow T_2$ denotes $T_1 \supseteq T_2$~\cite{cores}.}
\label{label_containment}
\end{figure}

\begin{definition}[Innermost truss]
Let $T_{k^*}[G]$ denote the innermost truss of $G$ , i.e., the non-empty $k$-truss of $G$ with the largest $k$. 
\end{definition}

\begin{lemma}
Given a temporal graph $G =(V,T,\tau)$, let $T_M$ be the set of all maximal span-trusses of $G$, and $T_{inner}$ = $\{T_{k^*}[G_{\Delta}] | \Delta \sqsubseteq T \}$  be the set of innermost trusses of all graphs $G_{\Delta}$. It holds that $T_M \subseteq T_{inner}$.
\label{lemma1}
\end{lemma}

\begin{proof}
Every $T_{k,\Delta} \in T_M$ is the innermost truss of the non-temporal graph $G_{\Delta}$: else, there would exist another truss $T_{k',\Delta} \neq \emptyset$ with $k' > k$, implying that $T_{k, \Delta} \not\in T_M$.
\end{proof}

\begin{lemma}
Given a temporal graph $G =(V,T,\tau)$, and three temporal intervals $\Delta = [t_s,t_e] \sqsubseteq T$, $\Delta' = [t_s - 1,t_e] \sqsubseteq T$, and $\Delta'' = [t_s,t_e + 1] \sqsubseteq T$. The innermost truss $T_{k^*}[G_{\Delta}]$ is a maximal span-truss of $G$ if and only if $k^* > \text{max}\{k',k''\}$ where $k'$ and $k''$ are the orders of the innermost trusses of $G_{\Delta'}$ and $G_{\Delta''}$, respectively.
\label{lemma2}
\end{lemma}

\begin{proof}
The "$\Rightarrow$" part comes directly from the definition of maximal span-truss (Definition~\ref{maximal_spantruss}): if $k^*$ were not larger than $\text{max}\{k',k''\}$, then $T_{k^*}[G_{\Delta}]$ would be dominated by another span-truss both on the order and on the span (as both $\Delta'$ and $\Delta''$ are super intervals of $\Delta$). For the "$\Leftarrow$" part, from Lemma~\ref{lemma1} and Proposition~\ref{prop_containment} it follows that $\text{max}\{k',k''\}$ is an upper bound on the maximum order of a span-truss of a super interval of $\Delta$. Therefore, $k^* > \text{max}\{k',k''\}$ implies that there cannot exist any other span-truss that dominates $T_{k^*}[G_{\Delta}]$ both on the order and on the span.
\end{proof}
\section{Efficient computation of maximal span-trusses}

We present our solution by first giving a na\"ive approach, and then by introducing three versions (Baseline, Streaming, Heuristic) that improve over the previous version.

\subsection{A na\"ive approach}
A first approach to solve the problem could be based on  Observation~\ref{lentezza}; namely, we can repeat the truss decomposition for every possible interval and then filter out non-maximal span-trusses.

\renewcommand{\algorithmicforall}{\textbf{forall}}
\begin{algorithm}
\caption{Na\"ive maximal span-trusses}
\label{naive}
\hspace*{\algorithmicindent} \textbf{Input:} A temporal graph $G =(V,T,\tau)$. \\
\hspace*{\algorithmicindent} \textbf{Output:} The set $T_M$ of all maximal span-trusses of $G$.
\begin{algorithmic}[1]
\State $\mathit{candidates} \gets \emptyset$
\State $T_M \gets \emptyset$
\ForAll {$t_s$ in [$0$, $1$, $...$ , $t_{max}$]} 
\ForAll {$t_e$ in $[t_s$, $t_s + 1$, $...$ , $t_{max}*]~|~E[t_s,t_e] \neq \emptyset$} 
\State $\Delta \gets [t_s, t_e]$
\State $\mathit{candidates}[\Delta] \gets \mathit{computeMaxTruss}(G_{\Delta})$
\EndFor
\EndFor
\State $T_M \gets$ maximal span-trusses from $\mathit{candidates}$
\end{algorithmic}
\end{algorithm}
Algorithm~\ref{naive} is trivially sound and complete since it iterates over every possible interval $\Delta$, extracts the maximal $k$-truss from $G_{\Delta}$ and saves it as a candidate element of $T_M$.

$T_M$ is constructed by filtering out non-maximal elements from $\mathit{candidates}$ and applying Definition~\ref{maximal_spantruss}.

\subsection{Baseline Algorithm}
As a baseline, we use a slightly better algorithm. This approach is similar to the baseline of the algorithm to mine span-cores~\cite{cores}. It exploits the containment properties we have proved before, which are shared between span-cores and span-trusses.

\renewcommand{\algorithmicforall}{\textbf{forall}}
\begin{algorithm}
\caption{Maximal span-trusses}
\label{baseline}
\hspace*{\algorithmicindent} \textbf{Input:} A temporal graph $G =(V,T,\tau)$. \\
\hspace*{\algorithmicindent} \textbf{Output:} The set $T_M$ of all maximal span-trusses of $G$.
\begin{algorithmic}[1]
\State $T_M \gets \emptyset$
\State $K'[t] \gets 0, \forall t \in T$
\ForAll {$t_s$ in $[0, 1, \ldots, t_{max}]$} 
\State $t^* \gets \max \{t_e \in [t_s,t_{max}]$ $|$ $E[t_s, t_e] \neq \emptyset\} $
\State $k'' \gets 0$
\ForAll {$t_e$ in $[t^*$, $t^* - 1$, $...$ , $t_s]$} 
\State $\Delta \gets [t_s, t_e]$
\State $\text{lb} \gets \max\{K'[t_e], k''\}$
\State $\InnermostTruss \gets \mathit{computeMaxTruss}(G_{\Delta})$
\State $k^* \gets$ order of $\InnermostTruss$
\If {$k^* > \text{lb}$}
\State $T_M \gets T_M \cup \{T\}$
\EndIf
\State $k'' \gets k^*$
\State $K'[t_e] \gets  \max \{K'[t_e], k''\}$
\EndFor
\EndFor
\end{algorithmic}
\end{algorithm}
Algorithm~\ref{baseline} works as follows.
It iterates over all the starting timestamps $t_s \in T$ in increasing order and, for each $t_s$, all the maximal span-trusses that have span starting in $t_s$ are identified. Proceeding in this way guarantees that a span-truss recognized as maximal will not be later dominated by another span-truss, since an interval $[t_s, t_e]$ can not be contained in another interval $[t_s',t_e']$ with $t_s < t_s'$.

To find all the maximal span-trusses having span starting in $t_s$, for any $t_s$ the algorithm identifies $t^{*} \geq t_s$, the maximum timestamp such that the edge set $E_{[t_s,t_e]}$ is not empty. Then, proceeding in decreasing order of $t_e$ and starting from $t_e = t^{*}$, all intervals $\Delta = [t_s, t_e]$ are considered (from the largest interval to the smallest interval).

The internal cycle computes the lower bound $\text{lb}$ (maximum between $K'[t_e]$ and $k''$) on the order of the innermost truss of $G_{\Delta}$ to be recognized as maximal. $K'$ is a map that maintains, for every timestamp $t \in [t_s,t^{*}]$, the order of the innermost truss of graph $G_\Delta'$ where $\Delta = [t_s-1, t]$ (i.e., $K'[t]$ stores what in Lemma~\ref{lemma2} is denoted as $k'$). $k''$ stores the order of the innermost truss of $G_\Delta''$ and $\Delta'' = [t_s , t_e + 1]$.

The selected truss is added to the set of the maximal span-trusses only if its order is larger than $\text{lb}$, then the values of $k''$ and $K'[t_e]$ are updated.



\begin{observation}
The worst-case time complexity of Algorithm~\ref{baseline} is $O(|T|^2 \times |E|^{1.5})$ since the $k$-truss decomposition (complexity $O(|E|^{1.5})$) is repeated for every $\Delta$. It is trivial to show that the number of possible intervals $\Delta$ is $O(|T|^2)$. Note that, since the output itself is potentially quadratic in $|T|$, it is not possible to improve over the $|T|^2$ factor in the computational complexity. 
\end{observation}

We outline now and discuss the operation of building the graph $(V, E_{\Delta})$ efficiently on both space and time; we follow the approach of~\cite{cores}.

Having a fixed timestamp $t_s \in [0,...,t_{\text{max}}]$, they propose the following reasoning which holds for every $t_s$. Let $E^-(t_e)=E_{[t_s,t_e]} \setminus E_{[t_s,t_e+1]}$ be the set of edges that are in $E_{[t_s,t_e]}$ but not in $E_{[t_s,t_e+1]}$, for $t_e \in [t_s,...,t^*-1]$. For each $t_s$, one can compute and store all edge sets $\{E^-(t_e)\}_{t_e \in [t_s, t^*-1]}$. These operations can be done in $O(|T|\times|E|)$ time, because every $E^-(t_e)$ can be computed incrementally from $E_{[t_e, t_e]}$ as $E^-(t_e) = \{(u,v) \in E_{[t_s, t_e]} | \tau(u,v,t_e+1)=0\}$.

For any $t_e$, $E_{[t_s, t_e]}$ can be reconstructed as $E_{[t_s, t_e+1]} \cup E^-(t_e)$, having previously computed $E_{[t_s, t_e+1]}$. Note that storing all $E^-(t_e)$ takes $O(|E|)$ space. That is why all $E^-(t_e)$ are stored and $E_{[t_s, t_e]}$ are reconstructed afterward instead of storing the latter, which would take $O(|T|\times|E|)$ space.

We use this approach in Algorithm~\ref{baseline}.

\begin{observation}
Since for any $t_e$, we reconstruct $E_{[t_s, t_e]}$ as $E_{[t_s, t_e+1]} \cup E^-(t_e)$, we are always adding new edges to the graph $G_{[t_s, t_e+1]}$ starting from an empty graph. This means we can exploit a streaming approach to solve the problem.
\label{streaming}
\end{observation}

\subsection{A Streaming Algorithm}


It is trivial to see that the Algorithm~\ref{baseline} repeats the truss decomposition in every possible interval. This means it also repeats the support computation, which for a single interval $\Delta$ has complexity $O(|E_{\Delta}|^{1.5})$ and it is the most expensive operation. Here we outline an algorithm to achieve better performance with regards to the support computation.

We can reframe the problem and think of it as a streaming problem, as stated in Observation~\ref{streaming}. Suppose we have computed the support for every edge active in the interval $\Delta^* = [t_s, t_e + 1]$. In the next step, we consider the interval $\Delta = [t_s, t_e]$ and so we are considering the graph $G_{\Delta}$ which is simply the graph $G_{\Delta^*}$ with a number of edges added, namely $E^-(t_e)$. We can study how the addition of these new edges changes the support of the edges of the old graph $G_{\Delta^*}$ and develop an algorithm that computes only the support of the edges in $E^-(t_e)$ and just \textit{updates} the support of the edges in $G_{\Delta^*}$. The updating part, without always recomputing, leads to a high speedup in the performance, as we will see in the next section.

After the update of the support of the edges, we can run the truss decomposition algorithm.

\renewcommand{\algorithmicforall}{\textbf{forall}}
\begin{algorithm}
\caption{Computing the support of every edge in $G_{\Delta}$ efficiently}\label{maximal}
\hspace*{\algorithmicindent} \textbf{Input:} A graph $G_{[t_s, t_e+1]} =(V,E_{[t_s, t_e+1]})$ with the support computed for every edge and a set $E^-(t_e)$ of edges to add to $G_{[t_s, t_e+1]}$ \\
\hspace*{\algorithmicindent} \textbf{Output:} A graph $G_{[t_s, t_e]}=(V,E_{[t_s, t_e+1]} \cup E^-(t_e))$ with the supports updated
\begin{algorithmic}[1]
\ForAll {$e \in E^-(t_e)$} 
\State add $e$ to $G_{[t_s, t_e+1]}$
\State let $(u, v) = e$
\ForAll {$w \in (\text{neighbours}(u) \cap \text{neighbours}(v))$} 
\State $\Sup(u,v)= \Sup(u,v)+1$
\State $\Sup(v,w)= \Sup(v,w)+1$
\State $\Sup(u,w)= \Sup(u,w)+1$
\EndFor
\EndFor
\end{algorithmic}
\end{algorithm}

\begin{observation}
If we use a map $M$, which maps a pair of vertices $(u,v)$ to $1$ if the edge exists in $G_{\Delta}$ at observation time or to $0$ if it does not exists, we can implement the intersection at step 4 by simply iterating over the neighbours of one of the two vertices and check in $O(1)$ if the remaining edge to form the triangle exists in the graph at observation time. Hence, the running time of this approach is bounded by $\sum_{(u,v) \in  E^-(t_e)} \text{min}\{\text{deg}(u),\text{deg}(v)\}$.
\end{observation}

\begin{figure}
\centering
\includegraphics[width=\linewidth]{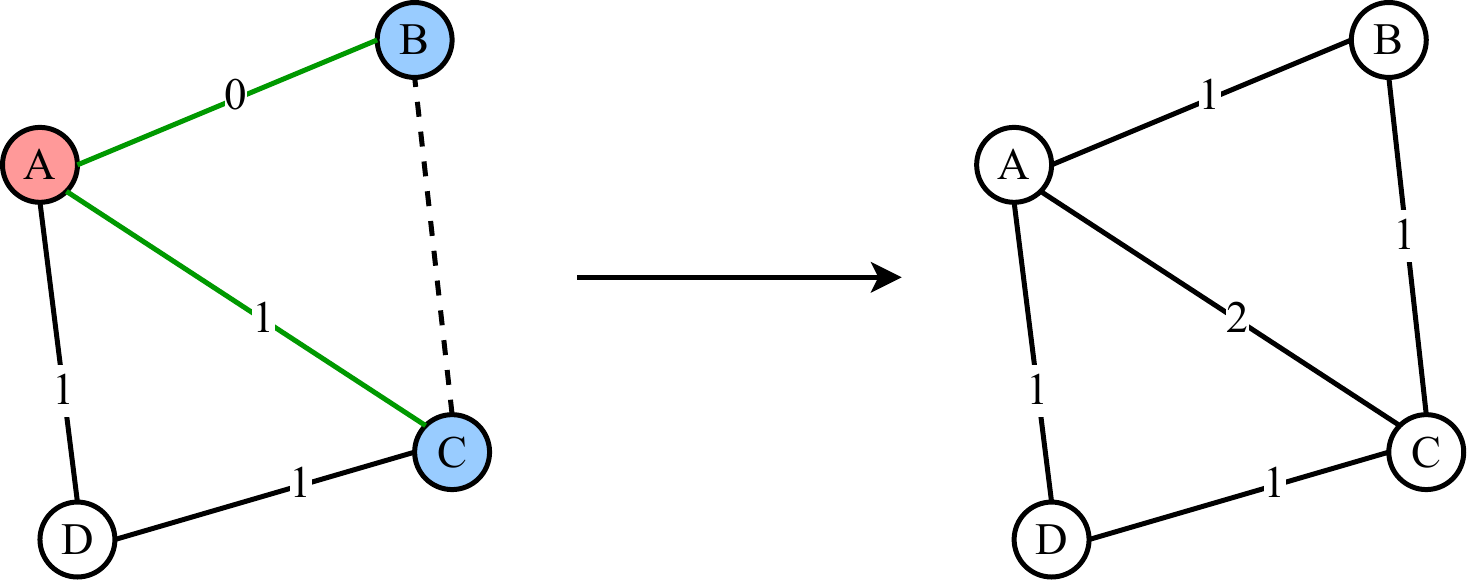}
\caption{In this example we show how the insertion of a new edge $(B,C)$ affects the supports of the other edges in the graph. The red vertex $A$ is the only vertex in $(\text{neighbours}(B) \cap \text{neighbours}(C))$, so we update the supports of $(A,C), (A,B)$ and of the new edge $(B,C)$. In fact we observe that $(B,C)$ forms a triangle with these edges, colored in green. On the right, we have the graph with the supports updated.}
\end{figure}

\subsection{Applying heuristics}
\label{heuristic}

It is worth mentioning that we still compute the truss decomposition in every graph $G_{\Delta}$. From Algorithm~\ref{baseline}, lines 11 to 14, we observe that a $k$-truss recognized as a maximal $k$-truss in a snapshot of a temporal graph will not always be recognized as a maximal span-truss.

\begin{observation}
If the order of the innermost-truss $I'$ of the graph $G_{[t_s, t_e]}$ is $k'$ and the order of the innermost-truss $I''$ of the graph $G_{[t_s, t_{e-1}]}$ is $k'$ then $I''$ is not a maximal span-truss.
\end{observation}

\begin{observation}
If the order of the innermost-truss $I'$ of the graph $G_{[t_s, t_e]}$ is $k'$ and the graph $G_{[t_s, t_{e-1}]}$ and the graph $G_{[t_s, t_e]}$ have the same number of edges with support greater than $k'-2$ then the order of $I''$ is $k'$.
\end{observation}

These two simple yet effective observations provide a minimal condition to avoid the computation of the truss decomposition in a snapshot of a temporal graph and lead to an improvement in the performance in particular datasets, as we will see in the next chapter.

\section{Evaluation}

\newcommand{\prosperloans}{\textsc{prosperloans}\xspace}
\newcommand{\lastfm}{\textsc{lastfm}\xspace}
\newcommand{\wikitalk}{\textsc{wikitalk}\xspace}
\newcommand{\wiki}{\textsc{wikipedia}\xspace}
\newcommand{\dblp}{\textsc{dblp}\xspace}
\newcommand{\stackoverflow}{\textsc{stackoverflow}\xspace}
\newcommand{\amazon}{\textsc{amazon}\xspace}

\begin{table}[b]
\small
\centering
\begin{tabular}{|p{1.9cm}p{0.5cm}p{0.5cm}p{0.5cm}p{0.8cm}p{1.9cm}|} 
 \hline
 Dataset & $|V|$ & $|E|$ & $|T|$ & window size (days) & domain\\ [0.5ex] 
 \hline\hline
 \prosperloans & 89k & 3M & 307 & 7 & economic\\ 
 \lastfm & 992 & 4M & 77 & 21 & co-listening\\
 \wikitalk & 2M & 10M & 192 & 28 & communication\\
 \dblp & 1M & 11M & 80 & 366 & co-authorship\\
 \stackoverflow & 2M & 16M & 51 & 56 & question\\ 
 &&&&&answering\\
 \wiki & 343k & 18M & 101 & 56 & co-editing\\
 \amazon & 2M & 22M & 115 & 28 & co-rating\\ 
 \hline
\end{tabular}
\caption{Description of the temporal graphs used for the experiments}
\label{table:1}
\end{table}

\paragraph*{Datasets}
We use eight real-world datasets recording timestamped interactions between entities\footnote{All datasets are made available by the KONECT Project (\url{http://konect.cc}), except for StackOverflow which is part of the SNAP Repository (\url{http://snap.stanford.edu}).}, as in~\cite{cores}. For each dataset, a window size is selected to build the corresponding temporal graph. Multiple interactions occurrinng between two entities during the same discrete timestamp are counted as one. The characteristics of the resulting graphs are reported in Table~\ref{table:1}.

\prosperloans represents the network of loans between the users of Prosper, a marketplace of loans between privates. \lastfm records the co-listening activity of the streaming platform Last.fm: two users are connected if they listened to songs of the same band during the same discrete timestamp. \wikitalk is the communication network of the English Wikipedia. \dblp is the co-authorship network of the authors of scientific papers from the DBLP computer science bibliography. \stackoverflow includes the answer-to-question interactions on StackOverflow. \wiki connects users of the Italian Wikipedia that co-edited a page within the same discrete timestamp. In the \amazon dataset, vertices are users, and edges represent the rating of at least one common item within the same discrete timestamp.

\paragraph*{Implementation}
The code\footnote{
\url{https://github.com/FraLotito/span_trusses}
} for the experiments has been implemented in C++11, compiled with g++ 5.4 and -O3 optimization, and run on a machine equipped with a 2,2 GHz CPU, 94GB RAM and Ubuntu 16.04.6 LTS (GNU/Linux 4.4.0-145-generic x86\_64).

\paragraph*{Results}
\begin{table}[t]
\small
\centering
\begin{tabular}{ |p{3.25cm}||p{3.25cm}|  }
 \hline
 Dataset & \# maximal span-trusses\\ [0.5ex] 
 \hline\hline
 \prosperloans & 293\\ 
 \lastfm & 1539\\
 \wikitalk & 466\\
 \dblp & 268\\
 \stackoverflow & 112\\
 \wiki & 1905\\
 \amazon & 303\\ 
 \hline
\end{tabular}
\caption{Number of maximal span-trusses in each dataset}
\label{table:2}
\end{table}

\begin{table}[t]
\small
\centering
\begin{tabular}{ |p{2cm}||p{1.5cm}|p{1.5cm}|p{1.5cm}|  }
 \hline
 Dataset & Baseline\newline(s) & Streaming\newline(s) & Heuristics\newline(s) \\ [0.5ex] 
 \hline\hline
 \prosperloans & 5 & 5 & 5\\ 
 \lastfm & 1318 & \textbf{1057} & 1109\\
 \wikitalk & 7497 & 818 & \textbf{336}\\
 \dblp & 513 & 112 & \textbf{85}\\
 \stackoverflow & 381 & 91 & \textbf{63}\\
 \wiki & 2447 & \textbf{1731} & 1837\\
 \amazon & 3025 & \textbf{2598} & 2607\\ 
 \hline
\end{tabular}
\caption{Experimental results}
\label{table:3}
\end{table}

Table~\ref{table:2} reports the number of maximal span-trusses that are present in the datasets. 

Table~\ref{table:3}, instead, shows the computing time for each of the datasets for the Baseline, Streaming and Heuristic algorithms. 
The table shows how computing the support of the edges in a streaming fashion improves the overall performance of the algorithm. We report a constant decrease in the time execution, with a peak with the \wikitalk dataset, which takes almost ten times less than the baseline.

The table also shows how our proposed heuristic to avoid unnecessary decompositions helps in reducing the time execution in some of the datasets, with a peak with the \wikitalk dataset which takes half the time with respect to our efficient algorithm. In some datasets, however, the heuristic comes with minimal overhead; we believe that it is worthwhile to use such version anyway, to exploit the more significant performance gain in the other cases.

\section{Related work}
The first and most obvious dense subgraph introduced to social network analysis is the clique, a subgraph in which every vertex is adjacent to every other vertex~\cite{Luce1949}. Computing cliques has several disadvantages. First, they are both too rare and too common: cliques of only a few members are frequently too numerous to be helpful, while larger cliques are too difficult to be found in real-world graphs. Second, no polynomial-time algorithm is known for this problem: this makes the enumeration of cliques impractical for moderate data sizes~\cite{Bron:1973:AFC:362342.362367}.

A number of generalizations and relaxations have been proposed to avoid the issues of rarity and tractability of cliques~\cite{alba1973graphtheoretic,plex, Mokken1979}. 

A well-known relaxation of the clique is the $k$-core decomposition~\cite{SEIDMAN1983269}. A $k$-core is a maximal subgraph in which each member is adjacent to at least $k$ other members. Unlike other clique generalizations, $k$-cores can be computed and listed in polynomial time. The disadvantage of $k$-cores is that they are too promiscuous and they can be of questionable utility.

The concept $k$-truss has been introduced as a compromise between the expensive-to-find and overly-numerous groupings provided by cliques, $k$-cliques, $k$-clubs, $k$-plexes on the one hand, and the easy-to-compute, few-in-number, but overly-generous $k$-cores on the other~\cite{Cohen}. In most real-world graphs, the maximum trussness is much lower than the maximum coreness, and the highest order truss is much denser than the highest-order core~\cite{Shin2018}. Figure~\ref{fig:corevstruss} highlights the differences between $k$-core and $k$-truss.

\begin{figure}
\centering
\includegraphics[width=\linewidth]{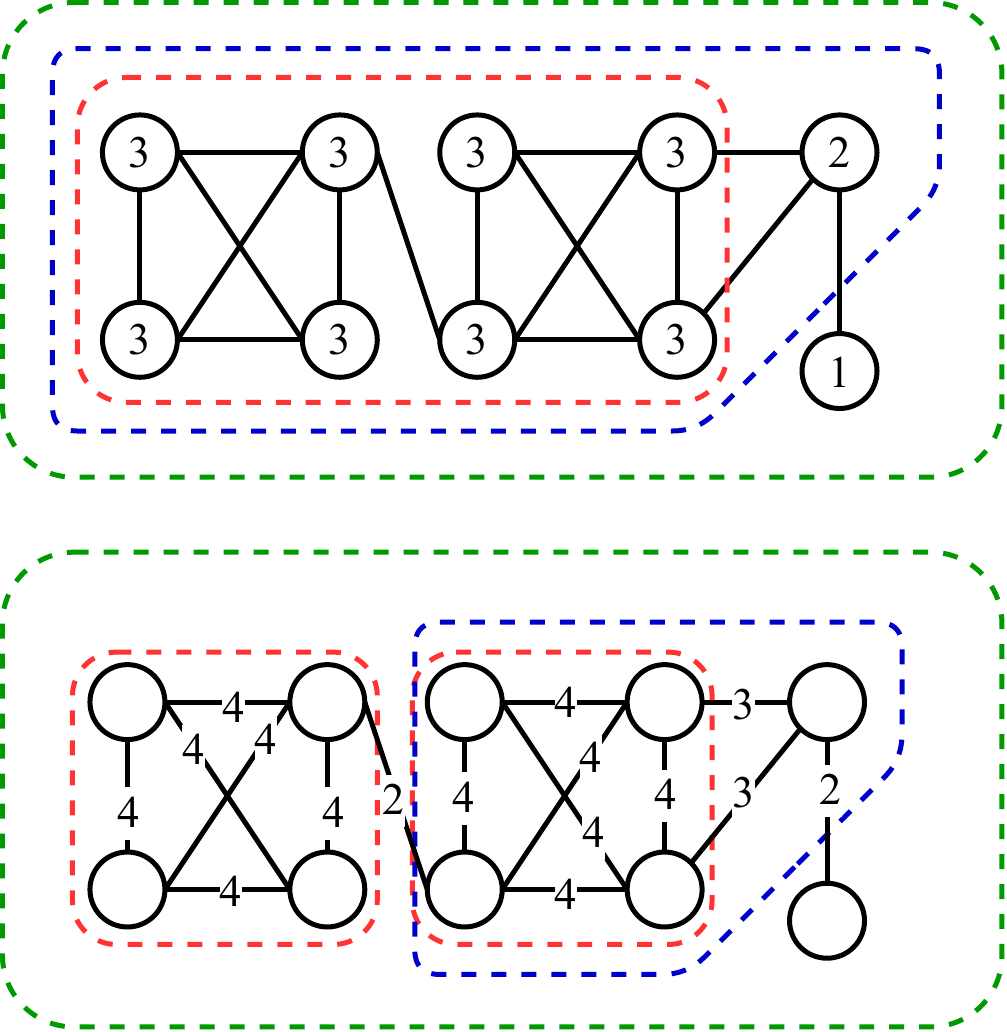}
\caption{Example of the differences between $k$-core (first picture) and $k$-truss decomposition (second picture)~\cite{mlg2019_5}. We highlight the coreness of every vertex in the first picture and the trussness of every edge in the second.}
\label{fig:corevstruss}
\end{figure}

Recently, there has been an increasing interest from the research community in generalizing cohesive structure concepts in a temporal setting.
Our work is directly inspired by the work of Galimberti et al.~\cite{cores} who generalized the concept of $k$-core and introduced the concept of \textit{span-core}. They also provided the corresponding algorithms to compute all the span-cores and to efficiently compute only the maximal ones (span-cores that are not dominated by any other span-core by both the coreness property and the span) in a temporal graph.

Other works related to ours include Semertzidis et al.~\cite{bff}, who introduced the problem of identifying a set of vertices that are densely connected in at least $k$ timestamps of a temporal network; Himmel at al.~\cite{himmel} and Viard et al.~\cite{viard}, who generalized the concept of clique in a temporal graph and proposed the respective listing algorithms; and Ma et al.~\cite{efficient_temporal_sub}, who a proposed a statistics-driven approach to find dense temporal subgraphs in large temporal networks.
\section{Conclusions}

In this paper, we have generalized the concept of $k$-truss to a temporal setting defining a structure called \textit{span-truss}, where each truss is associated with its span. We have developed both a na\"ive and an efficient algorithm to extract all the maximal span-trusses of a temporal graph, along with a heuristic to improve the running time in particular conditions. Finally, we have evaluated our proposals on a number of public datasets. 

In our future work, we plan to explore new heuristics to avoid the computation of the whole truss decomposition when not needed; for example, Burkhardt et al.~\cite{lb} summarized a number of properties and bounds that a $k$-truss must satisfy and which can be useful to avoid the computation of the decomposition when not needed.

\balance
\bibliographystyle{ACM-Reference-Format}
\bibliography{biblio}

\end{document}